\documentclass{sig-alternate}

\usepackage{times,amsmath,epsfig}
\usepackage{latexsym}
\usepackage{graphicx}
\usepackage{amsfonts}
\usepackage{amssymb}

\usepackage{subfigure}
\usepackage{algorithm}
\usepackage{algorithmic}
\usepackage{xspace}

\usepackage{multirow}

\begin{document}

%
\title{Query-driven Frequent Co-occurring Term Extraction over Relational Data using MapReduce}
\numberofauthors{5} 

\author{%
\alignauthor Jianxin Li \\
					\affaddr{Swinburne University of Technology}\\
					\affaddr{Melbourne, Australia}\\
					\email{jianxinli@swin.edu.au}
\alignauthor Chengfei Liu \\
					\affaddr{Swinburne University of Technology}\\
					\affaddr{Melbourne, Australia}\\
					\email{cliu@swin.edu.au}
\alignauthor Liang Yao \\
					\affaddr{Swinburne University of Technology}\\
					\affaddr{Melbourne, Australia}\\
					\email{liangyao@swin.edu.au}
\and 
\alignauthor Jeffrey Xu Yu \\
					\affaddr{Chinese University of Hong Kong}\\
					\affaddr{Hong Kong, China}\\
					\email{yu@se.cuhk.edu.hk}
\alignauthor Rui Zhou \\
					\affaddr{Swinburne University of Technology}\\
					\affaddr{Melbourne, Australia}\\
					\email{rzhou@swin.edu.au}
}

\date{18-23 March 2013}

\newdef{definition}{Definition}
\newtheorem{theorem}{Theorem}
\newtheorem{example}{Example}
\newtheorem{corollary}{Corollary}
\newtheorem{lemma}{Lemma}
\newtheorem{property}{Property}

\maketitle
\begin{abstract} 
In this paper we study how to efficiently compute \textit{frequent co-occurring terms} (FCT) in the results of a keyword query in parallel using the popular MapReduce framework. 
Taking as input a keyword query q and an integer k, an FCT query reports the k terms that are not in q, but appear most frequently in the results of the keyword query q over multiple joined relations. The returned terms of FCT search can be used to do query expansion and query refinement for traditional keyword search. Different from the method of FCT search in a single platform, our proposed approach can efficiently answer a FCT query using the MapReduce Paradigm without pre-computing the results of the original keyword query, which is run in parallel platform. In this work, we can output the final FCT search results by two MapReduce jobs: the first is to extract the statistical information of the data; and the second is to calculate the total frequency of each term based on the output of the first job. At the two MapReduce jobs, we would guarantee the load balance of mappers and the computational balance of reducers as much as possible. 
Analytical and experimental evaluations demonstrate the efficiency and scalability of our proposed approach using TPC-H benchmark datasets with different sizes.

 \end{abstract}

%
\section{Introduction}

Recently, analyzing and querying big data are attracting more and more research attentions, which can provide valuable information to company and personal customers.
For example, \cite{DBLP:conf/icde/ChandramouliGD12} combines a time-oriented data processing system with a MapReduce framework, which can allow users to perform analytics using temporal queries - these queries are succinct, scale-out-agnostic, and easy to write.
\cite{DBLP:journals/pvldb/CohenDDHW09} presents data parallel algorithms for sophisticated statistical techniques, with a focus on density methods, which enables agile design and 
flexible algorithm development using both SQL and MapReduce interfaces over a variety of storage mechanisms.
\cite{DBLP:journals/pvldb/AgrawalDA10,DBLP:journals/pvldb/Campbell11,DBLP:conf/pods/Chaudhuri12} summarize the state-of-the-art scalable data management
systems for traditional and cloud computing infrastructures. \cite{DBLP:journals/pvldb/AgrawalDA10} highlights update heavy and analytical workloads.
\cite{DBLP:journals/pvldb/Campbell11} introduces some important application examples for big data in real life.
\cite{DBLP:conf/pods/Chaudhuri12} describes six data management research challenges relevant for big data and the cloud.
Differently, in this paper our target is to address the problem of query-driven frequent co-occurring term extraction over big data, i.e., computing the frequent co-occurring terms in the results of a given keyword query. The returned frequent co-occurring terms, as the informative feedbacks, can be used to refine the original keyword query before the exact result set of the original query is retrieved.

Since traditional keyword search often assumes that the data sets should be loaded and processed in memory, it is not suitable to deal with keyword search over big data.
The challenges come from three points: (1) the ambiguity of keyword query limits the expressiveness of search intention and may lead to a large number of uninteresting results, which may make the users frustrated easily;
 (2) evaluating keyword query over big data requires scalable computational paradigm where a parallel platform is desirable;
 (3) generally, only simple index can be built for big data in practice due to the huge space cost.
 Due to the above challenges, there are only a few works to discuss the problem of keyword search over big data.
 \cite{DBLP:journals/vldb/QinYC11} addresses scalable keyword search on large data streams by pruning the unqualified tuples in a scalable method based on selection/semi-join strategy \cite{Bernstein:1981:USS:322234.322238}. \cite{DBLP:journals/pvldb/BaidRLDN10} addresses scalable keyword search over relational data by returning part of results, rather than the whole set of results, within the specified short time. 

By extracting frequent co-occurring terms of a given original keyword query, we can have two advantages naturally.
On the search engine side, it is more profitable to constrain users to a specific set of results by exploiting the frequent co-occurring terms of the issued keyword queries from big data if the users are also interested in the extracted terms, which can save lots of computational resources. On the user side, the users can easily discover the concepts that are closely associated with the given keyword set by extracting the frequent co-occurred terms from the big data, which is helpful for the users to easily understand their interesting information in the big data.

However, extracting such Frequent Co-occurring Terms (FCT) of a given keyword query is challenging today, as there is an increasing trend of applications being expected to deal with vast amounts of data that usually do not fit in the main memory of one machine, e.g., the Google N-gram dataset \cite{web:gramcorpus} and the GeneBank dataset \cite{web:genbank} that contains 100 million records with the total size of 416 GB. Applications with such datasets usually make use of clusters of machines and employ parallel algorithms in order to efficiently deal with this vast amount of data. For data-intensive applications, the MapReduce \cite{DBLP:journals/cacm/DeanG08} paradigm has recently received a lot of attention for being a scalable parallel shared-nothing data-processing platform. The framework is able to scale to thousands of nodes. In this paper, we use MapReduce as the parallel data-processing paradigm for extracting the frequent co-occurring terms with regards to a given keyword query over a big data.

We know that each keyword search result of a keyword query in relational database includes a set of tuples which are retrieved from a single relation or several joined relations, and contains all the given keywords of the keyword query.
Intuitively, given a keyword query and a big data, we can first compute the large number of keyword search results and then calculate the total frequencies for each term in the result set. At last, all the terms can be sorted by their frequencies and the top-$k$ frequent terms can be found. But this straightforward solution may make the feedbacks of the returned terms meaningless to the users because the highly time-consuming evaluation of keyword query over big data may delay the feedbacks greatly.

 To reduce the processing time, we propose a new FCT approach, which can avoid the procedure of direct keyword query evaluation using the idea of star algorithm in \cite{DBLP:conf/edbt/TaoY09}. Furthermore, our new approach can efficiently explore the statistical information from big data and compute the frequencies of the co-occurring terms using MapReduce.

The main contributions in this paper are summarized as follows.
\begin{itemize}
\item we propose a novel MapReduce-based approach to efficiently compute the query-driven frequent co-occurring terms over big data.
\item we propose two shuffling strategies to guarantee the load/ computational balance of mappers and reducers by considering both the uniformed data distribution and the uneven data distribution.
\item we conducted extensive performance studies to demonstrate the scalability of our proposed approach using TPC-H benchmark dataset.
\end{itemize}


The remainder of this paper is organized as follows. 
In Section~\ref{sec:preliminary}, we introduce the working procedure of MapReduce framework and an optimized multiway join in MapReduce in more details.
We define the problem of query-driven frequent term extraction (denoted as FCT search) in Section~\ref{sec:problemdefinition}. 
Section~\ref{sec:mrfct} firstly discusses the partition strategies for uniformed data distribution and uneven data distribution. And it then presents the procedures of computing frequent co-occurring terms of a query using MapReduce step by step. It lastly proves the correctness and completeness of the MapReduce-based FCT search. We provide the implementation algorithms of our approach in Section~\ref{sec:implementation}. 
Section~\ref{sec:experiments} presents the performance studies. 
Finally, we discuss related work in Section~\ref{sec:relatedwork} and conclude in Section~\ref{sec:conclusions}. 
 
\section{Preliminaries}\label{sec:preliminary}

\subsection{MapReduce Framework}

\begin{figure}[htbp]
  \centering
  \includegraphics[scale=0.55]{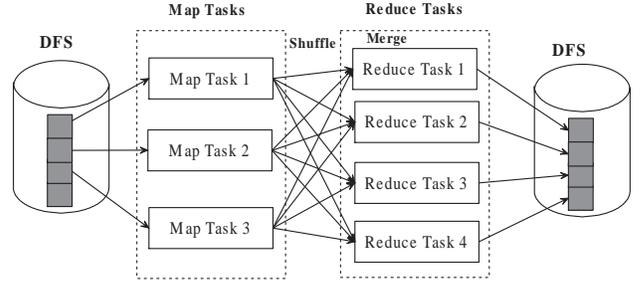}
  \caption{Overview of MapReduce}
  \label{fig:mrstructure}
  \end{figure}
  
 MapReduce~\cite{DBLP:journals/cacm/DeanG08} is a popular paradigm for data-intensive
parallel computation in shared-nothing clusters. Example
applications for the MapReduce paradigm include processing
crawled documents, Web request logs, etc. In the open source
community, Hadoop \cite{web:hadoop} is a popular implementation of
this paradigm. In MapReduce, data is initially partitioned
across the nodes of a cluster and stored in a distributed file
system (DFS). Data is represented as $(key, value)$ pairs.
The computation is expressed using two functions:

\begin{tabular}{llll}
map & $(k_1, v_1)$ & $\rightarrow$ & $list(k_2,v_2)$;\\
reduce & $(k_2, list(v_2))$ & $\rightarrow$ & $list(k_3, v_3)$.\\
\end{tabular}

Figure 1 shows the data flow in a MapReduce computation.
The computation starts with a map phase in which the
map functions are applied in parallel on different partitions 
of the input data. The $(key, value)$ pairs output by each
map function are hash-partitioned on the key. For each partition
the pairs are sorted by their key and then sent across
the cluster in a shuffle phase. At each receiving node, all
the received partitions are merged in a sorted order by their
key. All the pair values that share a certain key are passed
to a single reduce call. The output of each reduce function
is written to a distributed file in the DFS.
Besides the map and reduce functions, the framework also
allows the user to provide a combine function that is executed
on the same nodes as mappers right after the map
functions have finished. This function acts as a local reducer,
operating on the local $(key, value)$ pairs. This function allows
the user to decrease the amount of data sent through
the network. The signature of the combine function is:
combine $(k_2,list(v_2)) \rightarrow list(k_2, list(v_2))$.
Finally, the framework also allows the user to provide initialization
and tear-down function for each MapReduce function
and customize hashing and comparison functions to be used
when partitioning and sorting the keys. From Figure 1 one
can notice the similarity between the MapReduce approach
and query-processing techniques for parallel DBMS \cite{DBLP:journals/cacm/DeWittG92,DBLP:conf/sigmod/PavloPRADMS09}.

\subsection{Lagrangean Multipliers based Multiway Joins in MapReduce}


To learn how to optimize map-keys for a multiway join in \cite{DBLP:conf/edbt/AfratiU10}, 
let us begin with a simple example: the cyclic join
$R(A,B)$ $\Join$ $S(B,C)$ $\Join$ $T(A,C)$.
Suppose that the target number of map-keys is $k$. That
is, we shall use $k$ Reduce processes to join tuples from the
three relations. Each of the three attributes $A$, $B$, and $C$
will have a share of the key, which we denote $a$, $b$, and $c$,
respectively. We assume there are hash functions that map
values of attribute $A$ to $a$ different buckets, values of $B$ to $b$
buckets, and values of $C$ to $c$ buckets. We use h as the hash
function name, regardless of which attribute's value is being
hashed. Note that $abc = k$.

Consider tuples $(x, y)$ in relation $R$. Which Reduce processes need to know about this tuple? Recall that each Reduce process is associated with a map-key $(u, v, w)$, where
$u$ is a hash value in the range 1 to $a$, representing a bucket
into which $A$-values are hashed. Similarly, $v$ is a bucket in
the range 1 to $b$ representing a $B$-value, and $w$ is a bucket in
the range 1 to $c$ representing a $C$-value. Tuple $(x, y)$ from $R$
can only be useful to this reducer if $h(x)$ = $u$ and $h(y)$ = $v$.
However, it could be useful to any reducer that has these
first two key components, regardless of the value of $w$. We
conclude that $(x, y)$ must be replicated and sent to the $c$ different reducers corresponding to key values $(h(x), h(y), w)$,
where $1 \leq w \leq c$.
Similar reasoning tells us that any tuple $(y, z)$ from $S$ must be sent to the a different reducers corresponding to map-keys $(u, h(y), h(z))$, for $1 \leq u \leq a$. Finally, a tuple $(x, z)$ from $T$ is sent to the $b$ different reducers corresponding to map-keys
$(h(x), v, h(z))$, for $1 \leq v \leq b$.
 
This replication of tuples has a communication cost associated with it. The number of tuples passed from the Map processes to the Reduce processes is $rc + sa + tb$ where $r$, $s$, and $t$ are the numbers of tuples in relations $R$, $S$, and $T$, respectively.
Therefore, The optimization problem is to minimize the overall cost:
\begin{displaymath}
\text{Minimize }F(x) = rc + sa + tb\text{ subject to }abc = k 
\end{displaymath}
where $a$, $b$, and $c$ are the numbers of buckets of relations, and $k$ is the number of reduce tasks.

The method of Lagrangean multipliers serves us well.
That is, we start with the expression
$rc + sa + tb - \lambda(abc - k)$, and
take derivatives with respect to the three variables, a, b, and
c, and set the resulting expressions equal to 0. The result is
three equations:
$s = \lambda bc$ $=>$ $sa = \lambda k$;
$t = \lambda ac$ $=>$ $tb = \lambda k$;
$r = \lambda ab$ $=>$ $rc = \lambda k$.
If we multiply the left sides of the three equations and set
that equal to the product of the right sides, we get $rstk = \lambda^3k^3$ (remembering that $abc$ on the left equals k). We can now
solve for $\lambda = \sqrt[3]{rst/k^2}$. From this, the first equation $sa=\lambda k$ 
yields $a = \sqrt[3]{krt/s^2}$. Similarly, the next two equations yield
$b = \sqrt[3]{krs/t^2}$ and $c = \sqrt[3]{kst/r^2}$. When we substitute these
values in the original expression to be optimized,
$rc + sa + tb$, we get the minimum amount of communication
between Map and Reduce processes: $3\sqrt[3]{krst}$.



\section{Problem Definition}\label{sec:problemdefinition}
We consider that the database has $n$ tables $R_1$, $R_2$, ..., $R_n$,
referred to as the raw tables. Their referencing relationships
are summarized in a schema graph:

\begin{definition} (Schema Graph) The schema graph
is a directed graph $G$ such that (1) $G$ has $n$ vertices, corresponding to tables $R_1$, ..., $R_n$, respectively, and (2) $G$ has an edge from vertex $R_i$ to vertex $R_j$ ($1 \leq i \neq j \leq n$), if and only if $R_j$ has a foreign key referencing a primary key in $R_i$.
\end{definition}

\begin{definition} (Joining Network of Tuples)
A joining network of tuples $j_n$ is a tree of tuples where for each pair of adjacent tuples $t_i$, $t_j$ $\in j_n$, where $t_i \in R_i$, $t_j\in R_j$, there is an edge ($R_i, R_j$) in $G$ and ($t_i \Join t_j$) $\in$ ($R_i \Join R_j$).
\end{definition}

\begin{definition} (Keyword Query)
Given a keyword query, its result is the set of all possible joining networks of tuples that are both: (1) Total - every keyword is contained in at least one tuple of the joining network; (2) Minimal - we cannot remove any tuple from the joining network and still have a total joining network of tuples.
\end{definition}

As such, we can call such joining networks as \textit{Minimal Total Joining Networks of Tuples} (MTJNT) of the keywords in keyword query. Each MTJNT is a result instance of keyword query. To improve the efficiency of keyword query, we can group the tuples of each relation based on their contained query keywords.

\begin{definition} (Joining Network of Tuple Sets)
A joining network of tuple sets $J_n$ is a tree of tuple sets where for each pair of adjacent tuple sets $R_i^{K_i}$, $R_j^{K_j}$ in $J_n$, there is an edge ($R_i, R_j$) in $G$.
\end{definition}

Here, $R_i^{K_i}$ represents the set of tuples of relation $R_i$ where each tuple contains a partial query keyword set $K_i$ while $R_j^{K_j}$ represents the set of tuples of relation $R_j$ where each tuple contains a partial query keyword set $K_j$.

\begin{definition} (Candidate Network) 
Given a keyword query, a candidate network $C$ is a joining network
of tuple sets, such that there is an instance $I$ of the
database that has a MTJNT $M \in C$ and no tuple $t \in M$
that maps to a free tuple set $F \in C$ contains any keywords.
\end{definition}

As the example shown in \cite{DBLP:conf/vldb/HristidisP02}, for a keyword query
``Smith, Miller'', $J$ = ORDERS$^{Smith}$ $\Join$ CUSTOMER$^{\{\}}$ $\Join$ ORDERS$^{\{\}}$ is not a candidate network even though there is a MTJNT that belongs to $J$ because $J$ is subsumed by ORDERS$^{Smith}$ $\Join$ CUSTOMER$^{\{\}}$ $\Join$ ORDERS$^{Miller}$. Here, CUSTOMER$^{\{\}}$ and ORDERS$^{\{\}}$ denote free tuple sets.



\begin{definition} (Top-$k$ FCT Retrieval of Keyword Query) Given a keyword query $q$, a number $R_{max}$, and an integer $k$, a frequent co-occurring term (FCT) query returns the $k$ terms with the highest frequencies among all terms that (1) are not in $q$ and (2) in the results of the keyword query $q$ w.r.t. the maximal number $R_{max}$ of joined relations. 
\end{definition}

To the problem of FCT retrieval, a straightforward solution is to first solve the corresponding keyword query, and then extract the term frequencies. However, the solution would incur expensive cost because it needs to completely evaluate all the joins - the minimum total join networks (MTJNTs) of the corresponding keyword query. 
To reduce the computational cost, \cite{DBLP:conf/edbt/TaoY09} proposes a $star$ method to efficiently calculate the term frequencies without complete join evaluation. It first obtains all the candidate networks (CNs) of the keyword query by using the CN-generation algorithm in \cite{DBLP:conf/vldb/HristidisP02}. And then it computes the term frequencies for each CN. At last, all the computed term frequencies are summarized into the total term frequencies with regards to the FCT query over the data to be searched.

Let us use $h$ to represent the number of CNs. And a CN can be regarded as an algebraic expression, which retrieves a set of MTJNTs. We deploy MTJNT$(CN_i)$ to denote the set of MTJNTs resulting from executing $CN_i$ ($1 \leq i \leq h$) where we have MTJNT$(CN_i)$ $\bigcap$ MTJNT$(CN_j)$ = $\phi$ for any $1 \leq i \neq j \leq h$. That is to say, no MTJNT can be output by two CNs at the same time. Therefore, the keyword search result set can be defined as follows.

\begin{equation}
\text{KSResult}(q) = \bigcup _{i=1}^{h} \text{MTJNT}(CN_i).
\label{equ:ksresult}
\end{equation}
Let freq-CN$(CN_i, w)$ be the total number of occurrences for term $w$ in all the MTJNTs of MTJNT$(CN_i)$, or formally:
\begin{equation}
\text{freq-CN}(CN_i, w) = \sum_{\forall T\in \text{MTJNT}(CN_i)} Count(T, w).
\label{equ:cnfrequency}
\end{equation}
where $Count(T, w)$ is defined as the number of occurrences of $w$ in a single MTJNT $T$. Thus, the total frequency freq$(q, w)$ can be calculated as:
\begin{equation}
\text{freq}(q, w) = \sum_{i=1}^{h} \text{freq-CN}(CN_i, w).
\label{equ:queryfrequency}
\end{equation}
 
According to the above equation, the FCT retrieval (freq$(q, w)$) can be efficiently answered by alternatively calculating the term frequencies (freq-CN$(CN_i, w)$) of each candidate network $CN_i$. Specifically, freq-CN$(CN_i, w)$ can be calculated efficiently when $CN_i$ is a star candidate network where a vertex, called the root, connects to all the other vertices, called the leaves.  

Although the \textit{star} method is much better than the straightforward one, it is still expensive to compute FCT retrieval because (1) it needs to scan the data twice: scanning data for making statistical information and scanning data for computing the frequencies of terms in data; (2) it is a single-machine based approach, by which long-time processing will be incurred when the data to be processed is massive; (3) If $CN_i$ is not a star candidate network, it has to select some relations to do join operations exactly, which is used to make star-conversion. 
In this paper, we study the problem of FCT retrieval in the parallel environment, i.e., MapReduce framework, which can improve the performance greatly. In the following section, we mainly focus on the complex computation of FCT retrieval for star candidate networks. For the non-star candidate network, we can evaluate some relations to be selected using the repartition join strategy in MapReduce, which is easy to be implemented \cite{DBLP:conf/sigmod/BlanasPERST10}.  

\section{MapReduce-based FCT Search}\label{sec:mrfct}

Different from the \textit{star} method, we load and process the data in parallel using MapReduce. However, making the change is not trivial because MapReduce is a shared-nothing parallel data processing paradigm, and the performance of the FCT search would depend on whether the strategy of data partition is good or not. 
For each CN$_i$, we can get its relevant term frequencies by two MapReduce jobs. 

In this section, we first propose our optimal data partition strategy, by which the data can be distributed over the process nodes with the minimal duplications while it guarantees there is no communication cost among the process nodes in the period of evaluating the FCT search.
 And then, we describe the procedures of the two MapReduce jobs for aggregating the total term frequencies.
 At last, we analyze the properties of MapReduce-based FCT search approach.  

\subsection{Uniformed Distribution-based Shuffling \\ Strategy}\label{subsec:uniformed}
For the star-scheme join, we can still utilize the Lagrangean Multipliers based Join strategy to partition the data. For example, given a set of relations having R(A,B,C) $\Join$ S(A,E) $\Join$ T(B,F) $\Join$ P(C,G), the cost expression could be 
\begin{displaymath}
r + sbc + tac + pab 
\end{displaymath}
and the Lagrangean equations are: 
\begin{displaymath}
tac + pab = \lambda k, sbc + pab = \lambda k, sbc + tac = \lambda k 
\end{displaymath}
where $k = abc$. After making the pair-comparisons, we can get the transformed equations: $s/a = t/b = p/c$. Thus, the minimum-cost solution has shares for each variable proportional to the size of the dimension table in which it appears. That is to say, the map-keys partition the fact table into $k$ parts, and each part of the fact table gets equal-sized pieces of each dimension table with which it is joined. As a result, we can derive $a = \sqrt[3]{ks^2/tp}$, $b = \sqrt[3]{kt^2/sp}$ and $c = \sqrt[3]{kp^2/st}$. 

For instance, if we have 9 PC as processors, i.e. k=9, and each dimensional table contains 10,000 tuples, then each dimensional table can be splitted into 2 partitions, i.e., $S_0$, $S_1$, $T_0$, $T_1$, $P_0$, and $P_1$. Obviously, the fact table with the maximal size of $10,000^3$ ($10^{12}$) can be approximately distributed into the 8 processors only when the data in dimensional tables are distributed uniformly, i.e., one of 8 processors need to process $1.25*10^{11}$ joining operations maximally. 
The processors can be labelled as follows:
  \scalebox{0.8}{
  \begin{tabular}{|c|c|c|c|}  
    \hline
    000 & 001 & 010 & 011\\
    100 & 101 & 110 & 111\\
    \hline
    \end{tabular}
    }
	
	However, if one of the dimensional table contains skewed tuples, then some processors will become much hotter while the rest may be idle. Even if we can add more nodes to increase the scalability of the system, it does not solve the skew problem because all skewed tuples will still be sent to the hot processors. Following the above instance, if only the tuples in the first partition of dimensional table $S$ appear in the fact table $R$, then the half processors with labels $100, 101, 110, 111$ will be always idle. That is to say, each processor of $000, 001, 010, 011$ has to deal with $2.5*10^{11}$ joining operations maximally. This case often happens in star-scheme join operations because it is difficult to guarantee that data in all the dimensional tables are even distributed.
	To address the unbalance of computation, intuitively we can split the data into more bulks (small partitions) that can be distributed as evenly as possible. However, it is not easy to manage the scheduling of the big number of bulks, particularly when the bulk size is small.

\subsection{Uneven Distribution-based Shuffling \\ Strategy}\label{subsec:uneven}
 There are lots of work to process data skew in parallel joins in database systems. However, most of the existing work primarily focus on the join of two relations (R$\Join$S). Althoug their approaches can be applied to the join of multiple relations by repeated operations, the long-processing time would become challenging to some extent. Specifically, in this work our main problem is to address star-scheme join that consists of one fact table and more dimensional tables. From the study on the fact table, we can observe that the joined attributes in a fact table often contain some skewed tuples over one or several joined dimensions. For example, an airline can provide the booking service to travel agents and personal customers at the same time. Generally, a travel agent may book thousands of tickets per year, but a personal customer can book only a few tickets per year. If we treat the two types of customers equally, then the processors dealing with travel agents would be hot and the rest processors dealing with personal customers would be cooling (most of time, they are idle.) The output of hot processors will greatly decrease the overall performance of parallel join.   
 
  
 Consider a fact relation R(A,B,C) and two dimensional relations S(A,E) and T(B,F). Assume S can be splitted into three partitions: $S_0$, $S_1$, and $S_2$, and T is splitted into four partitions: $T_0$, $T_1$, $T_2$ and $T_3$. As such, we have 12 reduce tasks that need to be computed. We need to answer: how to distribute the 12 reduce tasks into the $k$ reducers (e.g., k=9)? We should be reminded that some reduce tasks does not produce joined results or only generate a few due to the data skew. To do this, we propose a cost model to evaluate the computational cost of each reduce task. Any formula for estimating the cost of a join could be used. Here, we chose the simple technique of estimating that
 
 \begin{math}
 c_{ij} = |R_{ij}|_{est} + |S_i|_{est} + |T_j|_{est} + |R_{ij} \Join S_i \Join T_j|_{est}
 \end{math}
 
 where $|R_{ij}|_{est}$ is an estimate of the number of R tuples mapped to the reduce task labelled as $ij$, $|S_i|_{est}$ is an estimate of the number of S tuples mapped to the reduce tasks labelled as $i?$ (the question mark means that all possible reduce tasks whose label starts with $i$), $|T_j|_{est}$ is an estimate of the number of T tuples mapped to the reduce tasks labelled as $?j$ (similarly, the question mark means that all possible reduce tasks whose label ends with $j$), $|R_{ij} \Join S_i \Join T_j|_{est}$ is an estimate of the number of tuples in $R_{ij} \Join S_i \Join T_j$. We compute this estimate of the size of $R_{ij} \Join S_i \Join T_j$ by assumping that the join attribute values in each join dimension $R_{ij}$ are uniformly distributed. Once this estimate for the cost of the joining of the reduce tasks are computed, any task scheduling algorithm can be used to try to balance the computational cost of reducers. In this work, we adopt a heuristic method to schedule the reduce tasks.
 For example, we have 5 reduce tasks with their estimated costs to be computed over two reducers. Then they will be assigned as shown in Figure~\ref{fig:unevenassign}.
  
  \begin{figure}[htbp]
  \centering
  \includegraphics[scale=0.55]{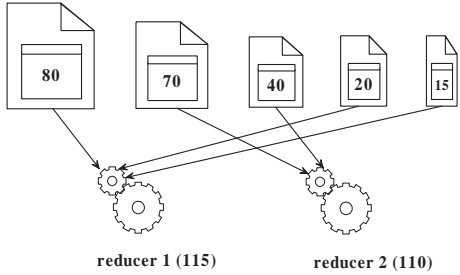}
  \caption{Example of Assigning Reduce Tasks}
  \label{fig:unevenassign}
  \end{figure}

 To efficiently estimate the cost, in this paper we employ \textit{Simple Random Sample} to choose the tuples from R, S and T. Although there are also other classific  sampling strategies, e.g., Stratified sampling and Systemantic sampling, they may introduce more workload. The detailed comparison of sampling strategies is out of the scope of this paper. 
 
After the reduce tasks are virtually placed to the reducers in balance, the master node will construct an allocation table to maintain the scheduling relationships between bulks and reducers. For example, if the reduce tasks $R_{00}$ and $R_{01}$ are grouped into the same reducer (e.g., $\#reducer = 1$) together, then the reducer with $\#reducer = 1$ will pull the tuple sets $R_{00}$ and $R_{01}$ from the corresponding mappers. At the same time, it also pull the tuple sets $S_0$, $T_0$ and $T_1$ from the corresponding mappers.



\subsection{Computing the Statistical Information at MapReduce$^{1st}$}

\subsubsection{Brief Procedure of Computing Statistical Information}
At MapReduce$^{1st}$, the map function gets as inputs the original tuples. For each tuple in dimensional relations, the function extracts the join attribute as the key $k_2$, and the texts of the rest attributes as the value $v_2$. To minimize the network traffic between the map and reduce functions, we use 
a combine function to aggregate the values with the same key output by the map function into a single value. And a number is appended to the aggregated value as the local frequency of the join attribute appearing in the current split. If the join attribute is the primary key of the 
corresponding relation, then the combine process can be skipped.

\begin{figure*}[t]
  \centering
  \subfigure[Star-CN]{\label{fig:starcn}
    \includegraphics[scale=0.6]{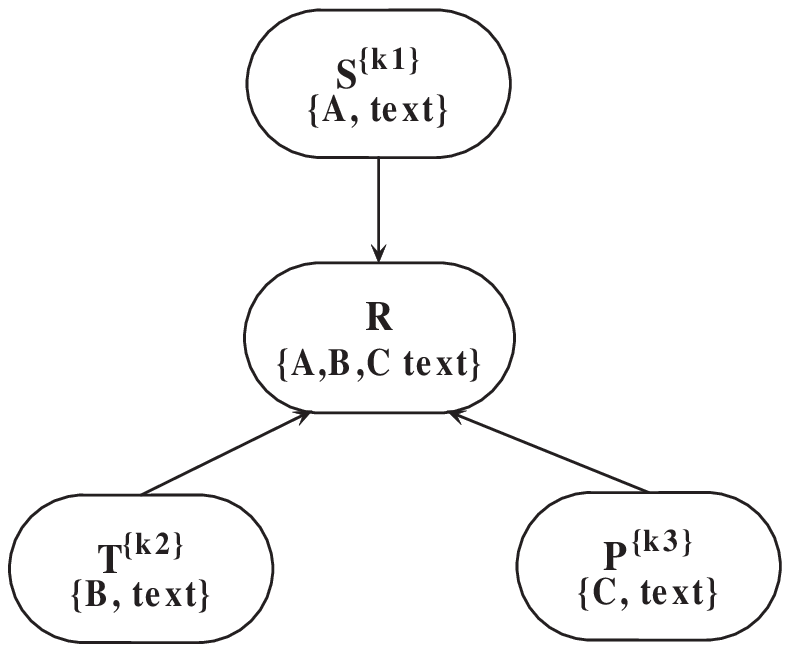}}
  \subfigure[$S^{\{k_1\}}$]{\label{fig:stable}
    \includegraphics[scale=0.6]{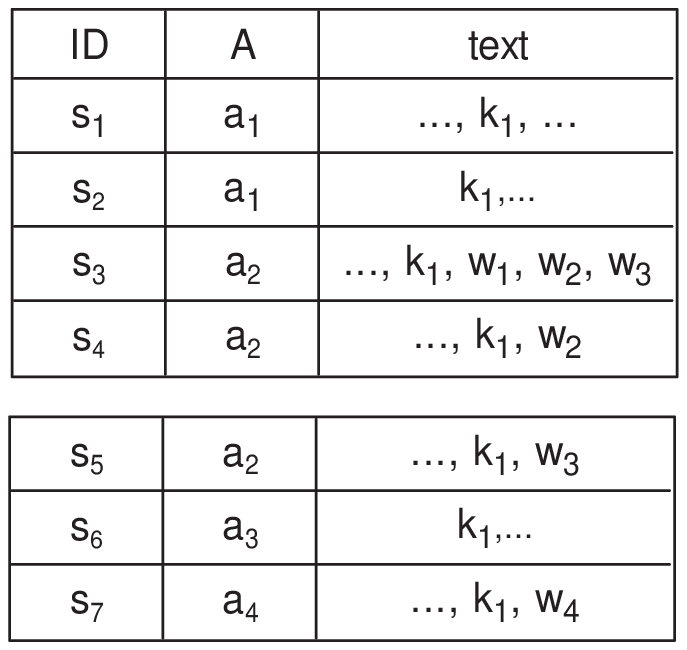}} 
  \subfigure[$T^{\{k_2\}}$]{\label{fig:ttable}
    \includegraphics[scale=0.6]{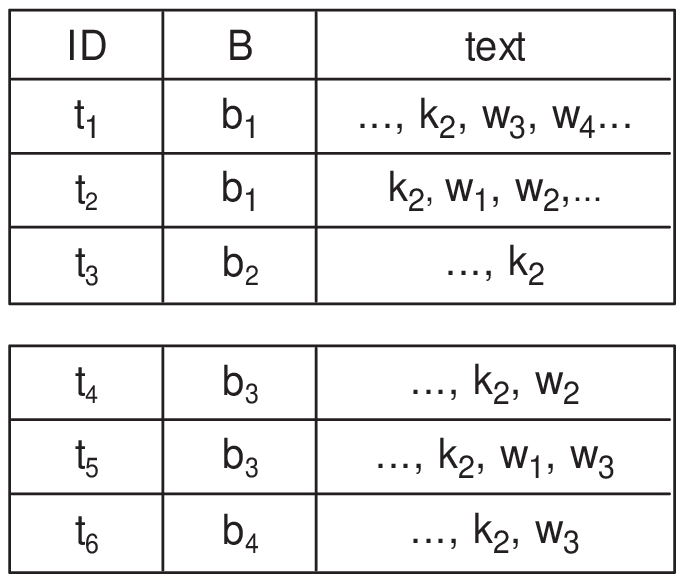}}  
 \\
  \subfigure[$P^{\{k_3\}}$]{\label{fig:ptable}
    \includegraphics[scale=0.6]{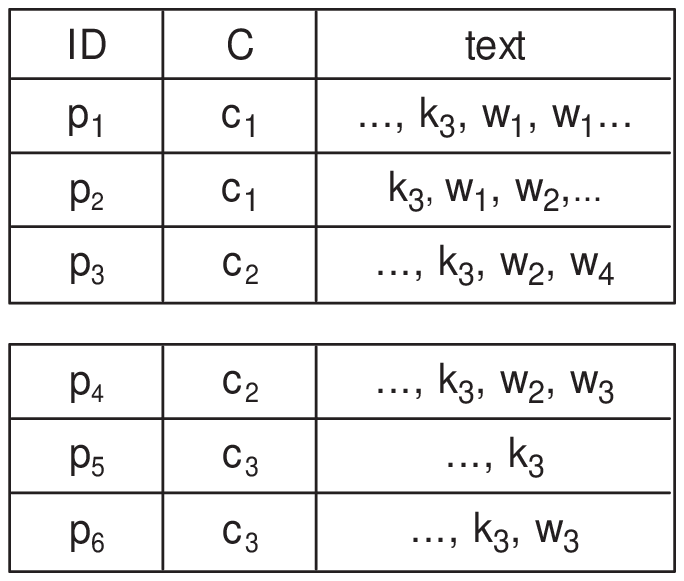}} 
  \subfigure[$R^{\phi}$]{\label{fig:rtable}
    \includegraphics[scale=0.6]{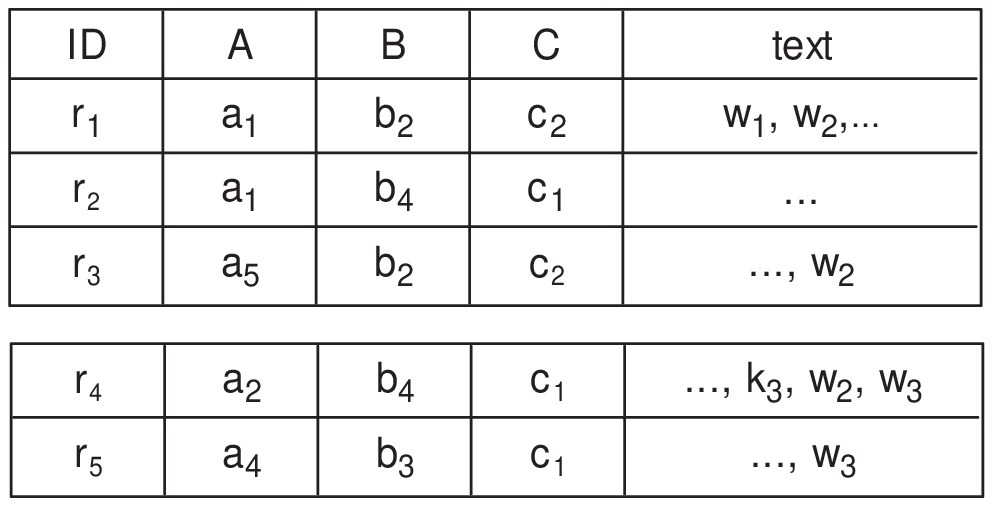}}    
         \caption{A running example for a query $\{k_1, k_2, k_3\}$}
  \label{fig:starjoin} 
\end{figure*}

Consider the example in Figure~\ref{fig:starjoin}. To make the load balanced, we assume that each table is splitted into two subtables with the equal size as much as possible, e.g., for $S^{\{k_1\}}$, the first partition size is $\left\lceil |S^{\{k_1\}}|/2 \right\rceil$ = 4 while the second partition is $|S^{\{k_1\}}| - \left\lceil |S^{\{k_1\}}|/2 \right\rceil$ = 3.
The map function can read the tuples $s_1$, $s_2$, $s_3$, $s_4$ and generate the corresponding key-value pairs for the first partition in Figure~\ref{fig:stable}. For $s_1$, the output key-value pair is ($a_1$, $''..., k_1, ...''$). For $s_2$, the key-value pair is ($a_1$, $''k_1, ...''$). Similarly, the rest two tuples generate key-value pairs taking $a_2$ as the key and the corresponding texts as their values respectively. By the combine function, we can aggregate the key-value pairs having the same keys, which can compress the data to be sent, e.g., the pairs of $s_1$ and $s_2$ can be aggregated into ($a_1$, $''..., k_1, ..., k_1, ... | 2''$) where the symbol $|$ is used to seperate the aggregated text information and the local frequency of the tuples (e.g., $s_1$ and $s_2$) with the same join attribute key (e.g., $a_1$); the rest two pairs can be aggregated into ($a_2$, $''..., k_1, w_1, w_2, w_3, ..., k_1, w_2 | 2''$). In this procedure, the query keywords (e.g., $k_1$) can be pruned directly, which does not affect the following calculation. But we keep it in the example in order to make the aggregation to be understood easily.

For the tuples in a fact relation, the map function extracts all the join attributes as the composite key $k_2$, and the texts of the rest attributes as the aggregated value $v_2$. Similarly, the combine function can be applied to aggregate the values of the tuples having the same set of join attributes in order to minimize the network traffic between the map and reduce functions. For example, the map function can output a key-value pair ($a_1|b_2|c_2$, $''w_1, w_2, ...''$) for the tuple $r_1$ in Figure~\ref{fig:rtable}. Similarly, we can transform the other tuples into key-value pairs.

Subsequently, the reduce function computes the statistical information for each fact tuple and its corresponding dimension tuples. Firstly, it pulls all the dimension tuples from the DFS files and counts the number of tuples sharing the same join attribute for each dimension relation, which are stored in a vector, denoted as \textit{num-array}. 
And then, it deals with the fact tuples one by one at each reducer. For each fact tuple, we need to probe the \textit{num-arrays} of its corresponding dimension relations for producing the cardinalities of the join attributes, which are stored in vectors, denoted as \textit{vol-arrays}. 
After all the fact tuples are processed at a reducer, it generates one \textit{vol-array} for each dimension relation and one \textit{vol-array} for the fact relation, which can be used to compute the frequencies of terms co-occurred with the query keywords at MapReduce$^{2nd}$.  

\subsubsection{Allocate Data to Reduce Tasks Evenly}
To understand the reduce function, we need to answer two questions. The first one is: how can we allocate the data to the reduce tasks as even as possible? 
Let's consider the discussion in Section~\ref{subsec:uniformed} again. If each dimension table is splitted into two partitions to be processed by reducers in parallel, then it will produce $2^3$ reduce tasks that are labelled as 
\scalebox{0.8}{
  \begin{tabular}{|c|c|c|c|}  
    \hline
    000 & 001 & 010 & 011\\
    100 & 101 & 110 & 111\\
    \hline
    \end{tabular}
    }. 
A possible way is to label the distinct join attributes with different numbers. And then, a hash function can be used to make the shuffle of the data. For example, there are four distinct attributes in column $A$. We can assign the numbers as $a_1 \rightarrow 1$, $a_2 \rightarrow 2$, $a_3 \rightarrow 3$, and $a_4 \rightarrow 4$. If we still want to split the data into two partitions, then the hash function can be designed as $h(a_i)=getNum(a_i) \text{ MOD } 2$
where $getNum(a_i)$ is used to get the number of the attribute $a_i$ in $A$ column. Based on the hash function, we have $h(a_1) = 1$, $h(a_2) = 0$, $h(a_3) = 1$, and $h(a_4) = 0$ respectively. Similarly, we can design hash functions for the join attributes in column $B$ and column $C$. 

According to the designed hash function, we can distribute the tuples in the fact table and dimension tables into different reduce tasks. For the key-value pair ($a_1|b_2|c_2$, $''w_1, w_2, ...''$) to be generated by the tuple $r_1$, it will be allocated to $h^A(a_1)h^B(b_2)h^C(c_2) = 100$. At the same time, for the key-value pairs with $a_1$, they will be distributed to the reduce tasks with 100, 101, 110 and 111, respectively. The key-value pairs with $b_2$ will be distributed to the corresponding reduce tasks with 000, 001, 100, and 101, respectively. The key-value pairs with $c_2$ will be distributed to the corresponding reduce tasks with 000, 010, 100, and 110, respectively. All the data in Figure~\ref{fig:starjoin} are allocated as shown in Figure~\ref{fig:allocation}, in which we only list the keys of the key-value pairs. In practice, both the keys and values of the key-value pairs will be allocated together.  

\begin{figure}[htbp]
  \centering
  \includegraphics[scale=0.6]{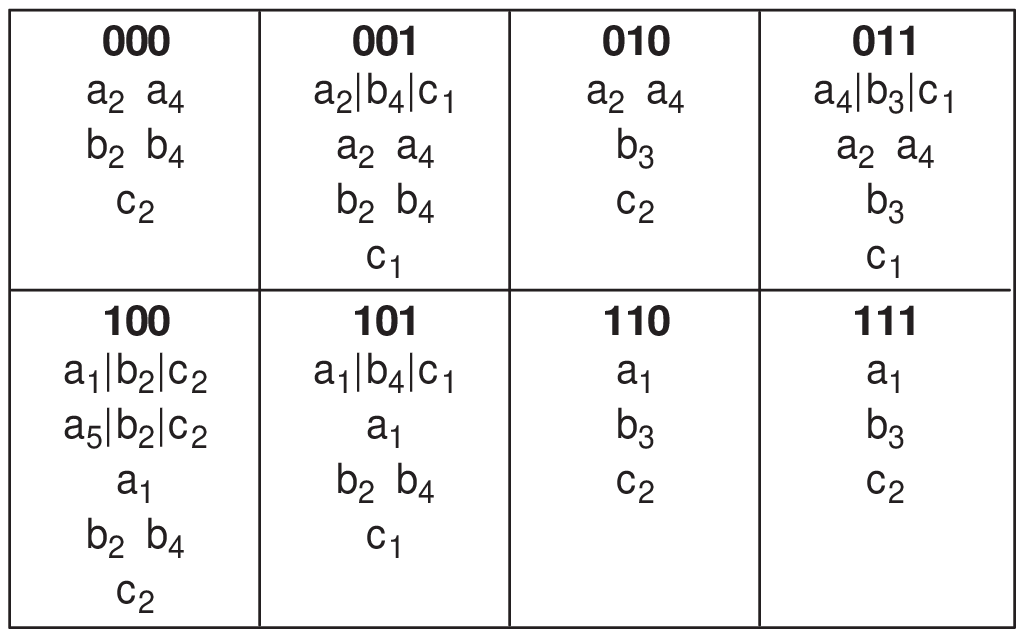}
  \caption{Demonstration of the Data Allocation}
  \label{fig:allocation}
  \end{figure}

\subsubsection{Allocate Reduce Tasks to Reducers Effectively}
The second question is: how can we allocate the reduce tasks to reducers effectively?
The naive method is to activate one reducer for each reduce task. But this method is often infeasible in real application. This is because we can generate different numbers of virtual reduce tasks for the same computational task while the number of available physical reducers is limited. Another possible method is to allocate the reduce tasks in a round-robin way. For example, the reduce tasks with the labels 000, 010, 100, and 110 can be computed at the first reducer while the rest tasks can be calculated at the second reducer. Although it intends to achieve the computation balance without consideration of data information, it often leads to unbalance of computation due to the data skew. 
In addition, different reduce tasks may have different computational workload. The reducers with the allocation of heavy workloads will take more processing time while the ones with the light workloads will take less processing time. Especially for the example in Figure~\ref{fig:allocation}, the reduce tasks with the labels 000, 010, 110, and 111 cannot make any contribution for the final results because no fact tuples appear in these tasks.
Therefore, these unuseful reduce tasks can be directly pruned without further computation, which can improve the performance a lot.  

After pruning the unuseful reduce tasks, the first reducer only need to process the reduce task (100) while the second reducer has to deal with the three reduce tasks (001, 011, 101). If we don't consider other information, then we can say that the computation over the two reducers have been balanced as much as possible because the first reducer needs to process two fact key-value pairs while the second reducer deals with three fact key-value pairs. However, if we do a simple sample over the reduce task (100), we can find that $a_5$ does not exist in the reduce task (100). Therefore, the fact key-value pair with key of $a_5|b_2|c_2$ can be pruned. At this moment, we can find that the workload ratio of the two reducers are $1:3$. To address the data skew, we can use the cost model in Section~\ref{subsec:uneven} to evaluate the cost of each reduce task and group them as even as possible. In this case, we can allocate the reduce tasks (100, 101) to the first reducer and the reduce tasks (001, 011) to the second reducer.  

\subsubsection{Detailed Procedure of Computing Statistical Information}
When the reducers pull all the relevant data from the corresponding mappers, it is time to build the local vol-arrays for the dimension and fact relations. For the reduce task (001), it gets the following key-value pairs: $(a_2|b_4|c_1,$ $''..., k_3,$ $w_2,$ $w_3'')$, $(a_2,$ $''...,$ $k_1,$ $w_1,$ $w_2,$ $w_3;$ $...,$ $k_1,$ $w_2$ $| 2'')$, $(a_2,$ $''...,$ $k_1,$ $w_3'')$, $(a_4,$ $''...,$ $k_1,$ $w_4'')$, $(b_2,$ $''...,$ $k_2'')$, $(b_4,$ $''...,$ $k_2,$ $w_3'')$ and $(c_1,$ $''...,$ $k_3,$ $w_1,$ $w_1,$ $...;$ $k_3,$ $w_1,$ $w_2,$ $...$ $| 2'')$. 
Firstly, we can build the num-arrays for the join attributes A, B and C, respectively. 

For the local num-array of $S^{\{k_1\}}$, we have

\scalebox{1}{
\begin{tabular}{|c|c|l|}  
    \hline
     attribute & num & text \\
     \hline
     $a_2$ & 3 & $''..., k_1, w_1, w_2, w_3; ..., k_1, w_2;$\\
     && $..., k_1, w_3''$\\     
     $a_4$ & 1 & $''..., k_1, w_4''$\\
    \hline
    \end{tabular}
}

For the local num-array of $T^{\{k_2\}}$, we have

\scalebox{1}{
\begin{tabular}{|c|c|l|}  
    \hline 
    attribute & num & text \\
     \hline
     $b_2$ & 1 & $''..., k_2''$ \\
     $b_4$ & 1 & $''..., k_2, w_3''$ \\
    \hline
    \end{tabular}
}

For the local num-array of $P^{\{k_3\}}$, we have 

\scalebox{1}{
 \begin{tabular}{|c|c|l|}  
    \hline
     attribute & num & text \\
     \hline
     $c_1$ & 2 & $''..., k_3, w_1, w_1, ...; k_3, w_1, w_2, ...''$\\    
    \hline
    \end{tabular}   
  }    
 
After we build the num-arrays for the local partitions, we need to process the fact key-value pairs one by one. Since the reduce task (001) only includes $(a_2|b_4|c_1, ''..., k_3, w_2, w_3'')$, the local vol-arrays of the fact and dimension tables can be built as follows.
   
For the local vol-array of $S^{\{k_1\}}$, we have

\scalebox{1}{
\begin{tabular}{|c|c|l|}  
    \hline
    attribute & volume & text \\
    \hline
     $a_2$ & 2 & $''..., k_1, w_1, w_2, w_3; ..., k_1, w_2;$ \\
     && $..., k_1, w_3''$\\      
    \hline
    \end{tabular}
  }
    
 For the local vol-array of $T^{\{k_2\}}$, we have

\scalebox{1}{
\begin{tabular}{|c|c|l|}  
    \hline 
    attribute & volume & text \\
     \hline    
     $b_4$ & 6 & $''..., k_2, w_3''$ \\
    \hline
    \end{tabular}
}

For the local vol-array of $P^{\{k_3\}}$, we have 

\scalebox{1}{
 \begin{tabular}{|c|c|l|}  
    \hline
     attribute & volume & text \\
     \hline
     $c_1$ & 3 & $''..., k_3, w_1, w_1, ...; k_3, w_1, w_2, ...''$\\    
    \hline
    \end{tabular} 
    }
    
 For the local vol-array of $R^{\phi}$, we have 

\scalebox{1}{
 \begin{tabular}{|c|c|l|}  
    \hline
     attribute & volume & text \\
     \hline
     $a_2|b_4|c_1$ & 6 & $''..., k_3, w_2, w_3''$\\    
    \hline
    \end{tabular}     
  }
    
Similarly, we can process the reduce tasks of 011, 100, and 101, respectively.
Since the dimension key-value pairs are often needed to be copied across different reduce tasks according to our adopted scheduling strategy, the by-product of the strategy is to avoid to re-pull the key-value pairs if they have been obtained by the reducers. By doing this, we can reduce the communication cost and guarantee the correctness of the results. 
For example, the reduce tasks 001 and 011 would be processed at the first reducer together. After we deal with the reduce task 001, the data information of $a_2$, $a_4$, $c_1$ are ready  at this reducer. For the reduce task 011, the reducer only needs to pull the necessary key-value pairs ($b_3$, $''..., k_2, w_2; ..., k_2, w_1, w_3 | 2''$) and ($a_4|b_3|c_1$, $''..., w_3''$). As such, only the local num-array of $T^{\{k_2\}}$ of the reduce task 001 needs to be updated by adding the key-value pair of $b_3$. For the updated local num-array of $T^{\{k_2\}}$, we have

\scalebox{1}{
\begin{tabular}{|c|c|l|}  
    \hline 
    attribute & num & text \\
     \hline
     $b_2$ & 1 & $''..., k_2''$ \\
     $b_3$ & 2 & $''..., k_2, w_2; ..., k_2, w_1, w_3''$ \\
     $b_4$ & 1 & $''..., k_2, w_3''$ \\
    \hline
    \end{tabular}   
    }
    
After processing the fact key-value pair ($a_4|b_3|c_1$, $''..., w_3''$), we can update the local vol-arrays of the fact and dimension tables as follows.
   
For the local vol-array of $S^{\{k_1\}}$, we have

\scalebox{1}{
\begin{tabular}{|c|c|l|}  
    \hline
    attribute & volume & text \\
    \hline
     $a_2$ & 2 & $''..., k_1,w_1,w_2,w_3;...,k_1,w_2;$ \\
           &   & $...,k_1,w_3''$\\ 
     $a_4$ & 4 & $''..., k_1, w_4''$ \\   
    \hline
    \end{tabular}
    }
    
 For the local vol-array of $T^{\{k_2\}}$, we have

\scalebox{1}{
\begin{tabular}{|c|c|l|}  
    \hline 
    attribute & volume & text \\
     \hline 
     $b_3$ & 2 & $''..., k_2, w_2; ..., k_2, w_1, w_3$\\   
     $b_4$ & 6 & $''..., k_2, w_3''$ \\
    \hline
    \end{tabular}
}

For the local vol-array of $P^{\{k_3\}}$, we have 

\scalebox{1}{
 \begin{tabular}{|c|c|l|}  
    \hline
     attribute & volume & text \\
     \hline
     $c_1$ & 5 & $''..., k_3, w_1, w_1, ...; k_3, w_1, w_2, ...''$\\    
    \hline
    \end{tabular} 
    }
    
 For the local vol-array of $R^{\phi}$, we have 

\scalebox{1}{
 \begin{tabular}{|c|c|l|}  
    \hline
     attribute & volume & text \\
     \hline
     $a_2|b_4|c_1$ & 6 & $''..., k_3, w_2, w_3''$\\  
     $a_4|b_3|c_1$ & 4 & $''..., w_3''$ \\ 
    \hline
    \end{tabular}   
}

Similarly, we can get the vol-arrays at the second reducer as follows.

For the local vol-array of $S^{\{k_1\}}$, we have

\scalebox{1}{
\begin{tabular}{|c|c|l|}  
    \hline
    attribute & volume & text \\
    \hline
     $a_1$ & 4 & $''...,k_1,...; k_1,...;$ \\           
    \hline
    \end{tabular}
    }
    
 For the local vol-array of $T^{\{k_2\}}$, we have

\scalebox{1}{
\begin{tabular}{|c|c|l|}  
    \hline 
    attribute & volume & text \\
     \hline 
     $b_2$ & 4 & $''..., k_2$\\   
     $b_4$ & 4 & $''..., k_2, w_3''$ \\
    \hline
    \end{tabular}
}

For the local vol-array of $P^{\{k_3\}}$, we have 

\scalebox{1}{
 \begin{tabular}{|c|c|l|}  
    \hline
     attribute & volume & text \\
     \hline
     $c_1$ & 2 & $''..., k_3, w_1, w_1, ...; k_3, w_1, w_2, ...''$\\ 
      $c_2$ & 2 & $''..., k_3, w_2, w_4; ...,k_3, w_2, w_3''$\\    
    \hline
    \end{tabular} 
    }
    
 For the local vol-array of $R^{\phi}$, we have 

\scalebox{1}{
 \begin{tabular}{|c|c|l|}  
    \hline
     attribute & volume & text \\
     \hline
     $a_1|b_2|c_2$ & 4 & $''w_1, w_2, ...''$\\  
     $a_1|b_4|c_1$ & 4 & $''...''$ \\ 
    \hline
    \end{tabular}   

}

\subsection{Computing the Term Frequency \\at MapReduce$^{2nd}$}

At MapReduce$^{2nd}$, we will output the final results - frequent co-occurrences with the given keyword query by utilizing the statistical information in \textit{vol-arrays} at MapReduce$^{1st}$. 

The map function takes as inputs the intermediate results of Map Reduce$^{1st}$, 
i.e., \textit{vol-arrays} consisting of three parts: \{join attribute, volume, text information\}. For each record in the \textit{vol-arrays}, we break the text information into token set by using any tokenization method and filter the stop words from the generated token set. For each distinct token, we can get its local frequency by counting the times of the token appearing in the filtered token set. And then, we generate the frequency of the token at the mapper by multiplying its local frequency and the volume of the record, which can be taken as the value $v_2$. The token is taken as the key $k_2$. 
At the reducer stage, the reduce function starts to compute the total term frequency. For a certain key, the reduce function pulls all the corresponding records with the key from all the mappers.

Let's take the term $w_1$ as an example to show the procedure of MapReducer$^{2nd}$. Assume there are two available mappers: the first mapper takes as inputs the output of the first reducer at Map Reducer$^{1st}$ and the second mapper takes as inputs the output of the second reducer at MapReducer$^{1st}$. For the first mapper, it scans each record in vol-arrays and generates the key-value pairs taking term as key and its cardinality as value. For the record $a_2$, it first computes the local frequency of $w_1$ as 1; and then it calculates the cardinality by multiplying the local frequency and the volume of the record, e.g., $1*2 = 2$; lastly it outputs a key-value pair as ($w_1$, 2). At the same time, the other key-value pairs of the terms in the record can be output. Similarly, the record $b_3$ outputs ($w_1$, $1*2=2$) and the record $c_1$ outputs ($w_1$, $3*5=15$). For the second mapper, it outputs the key-value pairs ($w_1$, $3*2=6$) by $c_1$ and ($w_1$, $1*4=4$) by $a_1|b_2|c_2$, respectively.

At the reducer stage, each reducer gets all the key-value pairs of the keys to be allocated to the reducer and adds the cardinalities of each term as the total frequency of the term. For $w_1$, its total frequency is $2+2+15+6+4=29$.     



\subsection{Propertities of MapReduce-based \\FCT Search}

    
\begin{theorem}\label{theorem:aggregation}(Aggregation Equal Transformation)  
The aggregation of the num-arrays and the vol-arrays to be built over independent reduce tasks is equal to those to be built over the aggregated data information of the independent reduce tasks.  
\end{theorem}  
 \begin{proof}
 For the num-arrays, if a key appears in a reduce task, then all the key-value pairs with the same key must appear in the reduce task. It says that the local frequency (num) of the key should be the global frequency (num) of the key in the original relations. In other words, the num of a key in a reduce task can be used to serve all the reduce tasks that contain the key. Therefore, the aggregation of the num-arrays over independent reduce tasks can be equivalently transformed to that we first aggregate the distinct key-value pairs of the independent reduce tasks and then compute the num-arrays over the aggregated data. Because the equivalent transformation of num-arrays holds, the vol-arrays can also be built by alternatively accessing the aggregated data of the independent reduce tasks.  
 \end{proof}
 
 Based on the equivalent transformation in Theorem~\ref{theorem:aggregation}, the statistical results of one reduce task can be used for another reduce task if both reduce tasks include the same keys. Therefore, two corollaries can be derived as follows.  
   
\begin{corollary}(Incremental Computation) The num-arrays and the vol-arrays can be incrementally built across reduce tasks.
\end{corollary}
  
 \begin{corollary}(Data Filtering) The reducers only need to pull the necessary data information that have not been seen.
\end{corollary}            

Since the derivations are easy to be understood, we do not provide their proofs in this paper. According to the above two corollaries, we can further improve the performance of our approach by 
\begin{itemize}
\item Reducing the communication cost due to the avoidance of the repeated data to be pulled;
\item Accelarating the computation of reducers because the reducers can start to work early for the existing data that have been ready;
\item Avoiding the computation from scratch by incrementally maintaining the computational results of the reduce tasks that have been processed. 
\end{itemize}

\begin{corollary} (Correctness and Completeness) MapReduce-based FCT Search can compute the term frequencies for a keyword query over big data correctly and completely.
\end{corollary}
\begin{proof}
According to the uniformed distribution-based shuffling strategy in Section~\ref{subsec:uniformed} or uneven distribution-based shuffling strategy in Section~\ref{subsec:uneven}, we can see that for each fact tuple, it will be sent to one reduce task, i.e., no duplicates across different reduce tasks. And for the fact and dimension partition data at each reduce task, they can be used to compute the term frequencies independently. Therefore, it guarantees the local correctness and completeness of MapReduce-based FCT Search for the partition data with regards to the reduce task.

Based on the aggregation equal transformation in Theorem~\ref{theorem:aggregation}, we can conclude the global correctness and completeness of MapReduce-based FCT Search because the aggregated results of all the reduce tasks can be equally transformed to compute the results over the aggregated data partitions (i.e., the original data).
\end{proof}

\section{Implementation of MapReduce-based FCT Search}\label{sec:implementation}

In the above sections, we have introduced the concepts of our MapReduce-based FCT search approach. Now we present its implementation, which includes the functions Map(), Reduce() and getPartition() of MapReduce$^{1st}$ and the brief description of Map Reduce$^{2nd}$, respectively.

\begin{algorithm}[ht]
  \caption{Map(key, a record) at MapReduce$^{1st}$} \label{algo:map} 
    
    \begin{algorithmic}[1]
    
   		\STATE key$_{new}$ = getJoinAttribute(key, the record);
    	\IF{Type(key$_{new}$) is identified as a dimension key}    	
    			
    			\STATE indexPos = getJoinAttrPosition(Type(key$_{new}$), joinAttrTypeSet[]);
    			
    			\FOR{(i=1; i$<=$ numDuplicates; i++)}\label{line:copybegin}
    					
    					\STATE cPartition = Integer.toX-naryString(i);
    					\STATE cPartition.insertBefore(`*', indexPos);
    					
    					\STATE Value$_{new}$ = getValue(key, the record);
    			
    					\STATE Emit(pair(indexPos,key$_{new}$), pair(cPartition,Value$_{new}$));
    			\ENDFOR \label{line:copyend}   			
    		
    	\ELSE
    	
    			\STATE keyset = getJoinAttribute(key, the record); \label{line:factnocopy1}
    			
    			\STATE priority = $\sum$getJoinAttrPosition(Type(key $\in$ keyset), joinAttrTypeSet[]);
    			\STATE key = keyset.toString(`\textbar');
    			\STATE Value$_{new}$ = getValue(keyset, the record);
    			\STATE Emit(pair(priority,key), Value$_{new}$);   \label{line:factnocopy2} 			
    	\ENDIF 
     \end{algorithmic}
  \end{algorithm}

 In Algorithm~\ref{algo:map}, we show the procedures in Map() at MapReduce$^{1st}$. When a dimension tuple is read, it first extracts as the new key the join attribute and generates as the value a string by combining the contents of the rest attributes. And then, it tags the key with a number where we use the position of its corresponding attribute column in fact relation. By doing this, we can guarantee at each reducer, the data belonging to the same dimension relation will be collected together. And it also tags the value with the partition to be copied, which is used to implement the multiway join based data partition, as shown in Line~\ref{line:copybegin}-Line~\ref{line:copyend}. For a fact tuple, the map function tags the key with the sum of the position numbers of their corresponding attribute columns in fact relation, which guarantees that all the fact tuples should arrive after all the dimension tuples at a reducer. Different from processing the dimension tuples, we don't need to tag the values because each fact tuple will be sent to only one partition as shown in Line~\ref{line:factnocopy1}-Line~\ref{line:factnocopy2}.

 \begin{algorithm}[ht]
  \caption{getPartition(key, value, numReduceTasks)  at MapReduce$^{1st}$} \label{algo:getpartition} 
    
    \begin{algorithmic}[1]
    
    	\STATE \{Comments: compute the number numDimPartitions of partitions for dimension relation from numReduceTasks\}; \label{line:numdimpartition}
    	\IF{key.second() does not contain `\textbar', i.e.,a dimension tuple}
					\STATE numPartition = (key.second().hashCode() $\&$ Integer.MAX\_VALUE) $\%$ numDimPartitions; \label{line:dimassign}    
		    	
		    	
		    	\STATE cPartition = value.first().replace(`*', numPartition); \label{line:cpartition1}
		    	\RETURN cPartition.toDecimal() as the number of partition;\label{line:cpartition2}
		  	
		  \ELSE  	
		  		\STATE new a string str=`';
    			\STATE keys[] = key.second().split('\textbar'); \label{line:factkey1}
    			\FOR{int i=0; i$<$keys.length; i++}   					
    				
    					\STATE str += (keys[i].hashCode() $\&$ Integer.MAX\_VALUE) $\%$ numDimPartitions;
    			\ENDFOR
    			\RETURN str.toDecimal() as the number of partition;\label{line:factkey2}
    	\ENDIF   			 		
    	
     \end{algorithmic}
  \end{algorithm} 
  
 For adapting to the multiway join based data partition, Algorithm~\ref{algo:getpartition} redesigns the function getPartition() of Hadoop. According to the specified number of reduce tasks, i.e., numReduceTasks, Line~\ref{line:numdimpartition} is used to calculate the number (denoted as numDimPartition) of partitions for each dimension relation using the derived equations, e.g., $a = \sqrt[3]{ks^2/tp}$, $b = \sqrt[3]{kt^2/sp}$ and $c = \sqrt[3]{kp^2/st}$ in Section~\ref{subsec:uniformed}.
 If the key-value pair comes from a dimension relation, we can compute the local partition number, with regards to the dimension relation, to be allocated by the key, as shown in Line~\ref{line:dimassign}. And then, it will be used to compute the global partition number by combining it with the partition numbers of the other dimension relations, as shown in Line~\ref{line:cpartition1}-Line~\ref{line:cpartition2}. Similarly, we can process the key-value pairs coming from fact relation. Differently, the key is often a composite key that consists of multiple single keys. Therefore, we need to first calculate the local partition number for each single key and then transform the set of local partition numbers into a global partition number, as shown in Line~\ref{line:factkey1}-Line~\ref{line:factkey2}.

\begin{algorithm}[ht]
  \caption{Reduce(key, a record) at MapReduce$^{1st}$} \label{algo:reduce} 
    
    \begin{algorithmic}[1]
    
    		\STATE \{Comments: the dimension tuples are always processed before fact tuples\};
    		\IF{key.second() is identified as a dimension key} \label{line:getnum1}
		    		
		    		\IF{hash$^{i}_{value}$.contains(key.first())}
		    		
		    				\STATE hash$^{i}_{num}$(key.first())=hash$_{num}$(key.first())+1;
		    		\ELSE
		    				\STATE hash$^{i}_{value}$(key.first())=value.second();
		    				\STATE hash$^{i}_{num}$(key.first())=1;
		    		\ENDIF \label{line:getnum2}
		    \ELSE
		    		\STATE \{Comments: all the dimension tuples belonging to this partition have arrived\};
		    		
		    		\STATE keys = key.first().split(`\textbar'); \label{line:volume1}
		    		
		    		\STATE num$_i$ = hash$^{i}_{num}$(keys[i]);
		    		
		    		\IF{any num$_i$ $\neq$ 0}
		    			
		    				\STATE hash$^{f}_{value}$(key) = value;
		    				\STATE hash$^{f}_{vol}$(key) = $\prod$ num$_i$;
		    				
		    				\STATE hash$^{i}_{vol}$(keys[i]) += $\prod_{j \neq i}$ num$_j$; 
		    		\ENDIF \label{line:volume2}
		    		
		    \ENDIF   
   	 
   	 		\STATE \{Comments: generate inputs for the reducer at MapReduce$^{2nd}$\};
   	 		
   	 		\FOR{each item $x$ in hash$^{f}_{vol}$ or hash$^{i}_{vol}$} \label{line:output1}
   	 		
   	 				\FOR {each word $w$ in hash$^{\text{f or i}}_{value}$($x$)}
   	 				
   	 						\STATE Emit($w$, hash$^{\text{f or i}}_{vol}$($x$));
   	 				\ENDFOR
   	 				
   	 		\ENDFOR \label{line:output2}
     \end{algorithmic}
  \end{algorithm} 
  
  Algorithm~\ref{algo:reduce} can be divided into three stages. At the first stage in Line~\ref{line:getnum1}-Line~\ref{line:getnum2}, it calculates the total number of dimension records with the same join attribute as key. At the second stage in Line~\ref{line:volume1}-Line~\ref{line:volume2}, it calculates the volume for each fact tuple using $\prod$ num$_i$ and the total volumes for each join attribute using $\prod_{j \neq i}$ num$_j$, respectively. At the third stage in Line~\ref{line:output1}-Line~\ref{line:output2}, it generates the intermediate results that will be taken as the inputs of the reducer at MapReduce$^{2nd}$. Since the tag of the fact keys is always larger than that of the dimension keys, any dimension relation has the higher priority than the fact relation. Therefore, we guarantee that the three stages can be processed in a stable sequence.

    
    				

    	
  
The rest work is similar to the classific example of word count using MapReduce. We take as the inputs the intermediate results of the reducer at MapReduce$^{1st}$. And we calculate the total frequencies of each term. 
After that, the merge-sort operation is applied to the outputs of the reducers at MapReduce$^{2nd}$. As such, the top-$k$ frequent words or terms will be found.
  


\section{Experiments}\label{sec:experiments}
In this section, we study the performance of our proposed MapReduce-based FCT search approach. All experiments were performed on a 9-machine cluster running Hadoop 1.0.3 \cite{web:hadoop} at SwinCloud platform~\footnote{hadoop.ict.swin.edu.au}.
One machine served as the Head Node running CentOS-5 Linux with 500 GB hard disk allocated as DFS storage. The Head Node also serves as the Name Node and JobTracker at the same time. While the other 8 machines as Worker Nodes are the general PC with 1 GB RAM, which can be used for Map and Reduce tasks. And each Worker Node is configured to run one map and one reduce task concurrently. The distributed file system block size is set to 64MB. Only the Head Node takes the role of storage node for the DFS. All the machines are connected via a Gigabit-Ethernet network. 

\subsection{Selection of Dataset and Keyword Queries}
As TPC-H \cite{web:tpch} is the most widely used big data benchmark in MapReduce study, e.g., \cite{DBLP:conf/sigmod/LinACOW11,DBLP:journals/tkde/JiangTC11,DBLP:conf/icde/YangYTM10}, we generate a set of datasets with different sizes. In order to demonstrate the performance of the multiway-join in MapReduce, we directly link the PART relation and the SUPPLIER relation to the relation LINEITEM, by which the relation LINEITEM is taken as the fact relation while the relations PART, SUPPLIER and ORDERS are considered as the dimension relations. In addition, The original TPC-H Schema can be seen at \cite{web:tpch}. 

\begin{figure}[htbp]
  \centering
  \subfigure[Star-Type]{\label{fig:querytype1}
    \includegraphics[scale=0.5]{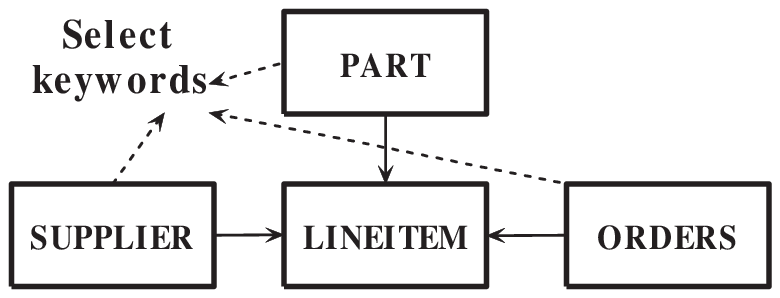}}
 \hskip 0.05in
  \subfigure[Chain-Type]{\label{fig:querytype2}
    \includegraphics[scale=0.5]{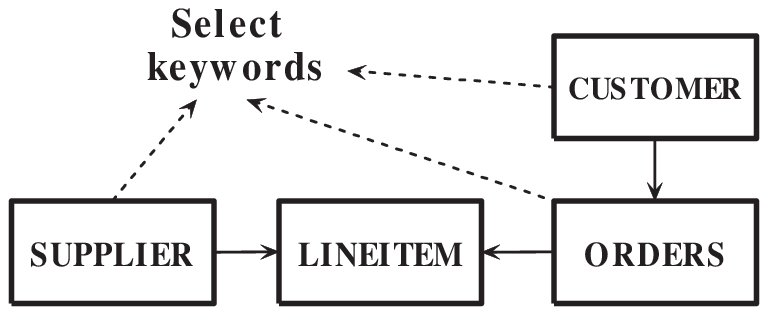}} 
   \\     
  \subfigure[Mix-Type]{\label{fig:querytype3}
    \includegraphics[scale=0.5]{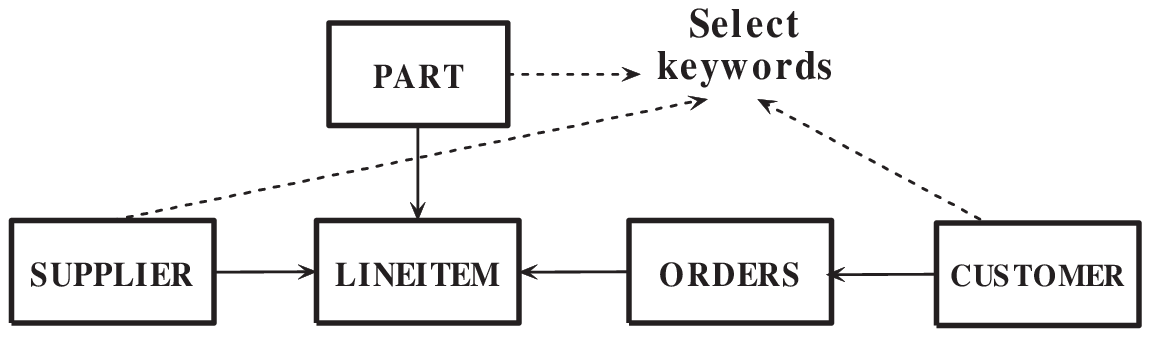}} 
  \caption{Designed Query Types}\label{fig:querytype} 
\end{figure} 

To test the stability of our approach, we design three types of keyword queries as shown in Figure~\ref{fig:querytype}. In order to guarantee the result sets of the generated keyword queries are not empty, we adopt the following steps to generate keyword queries and record their experimental results. Let's take the type$_{1}$ in Figure~\ref{fig:querytype1} as an example: 
\begin{itemize}
\item run the structured queries and output the results as three bags of texts, e.g., for star-type, we have: select bag(p*), bag(s*), bag(o*) from part p, supplier s, lineitem li, orders o where p.partkey = li.partkey \& s.suppkey = li.suppkey \& o.orderkey = li.orderkey; 
\item choose the terms from different bags as the query keywords; 
\item For each keyword query type, we generate 3 batches of keyword queries where each batch contains 10 random keyword queries. 
\end{itemize}
In the following study, the average performance of each batch of keyword queries are used to make comparison, e.g., $Q_1$, $Q_2$, and $Q_3$ represent the three batches for Type$_1$; $Q_4$, $Q_5$, and $Q_6$ represent the three batches for Type$_2$; $Q_7$, $Q_8$, and $Q_9$ represent the three batches for Type$_3$.         

To illustrate the advantages of parallel platforms, we first run the FCT search of the query batch Q$_1$ over 1GB dataset in single machine platform. Although the single machine has 4GB RAM, 500GB hard disk and it does not need shuffle operations, it still consumes about 4.5 mins to complete the FCT search of Q$_1$, which is much expensive than that (about 1.83 mins) of our parallel platform consisting of 8 worker nodes. Therefore, the following experimental studies are only focused on the evaluation of our approach in the parallel platform.  

\begin{figure*}[ht]
\hfill
\begin{minipage}[t]{.3\textwidth}
\begin{center}
 \epsfig{file=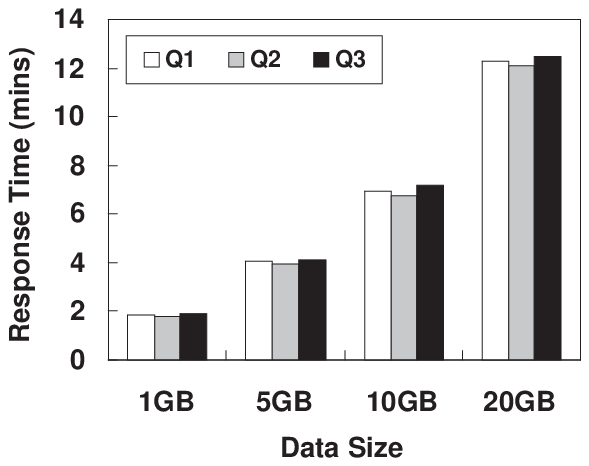, scale=0.85}
  \caption{Response Time of Star-Type Queries}
  \label{fig:startime}
  \end{center}   
\end{minipage}
\hfill 
\begin{minipage}[t]{0.3\textwidth}
\begin{center}
 \epsfig{file=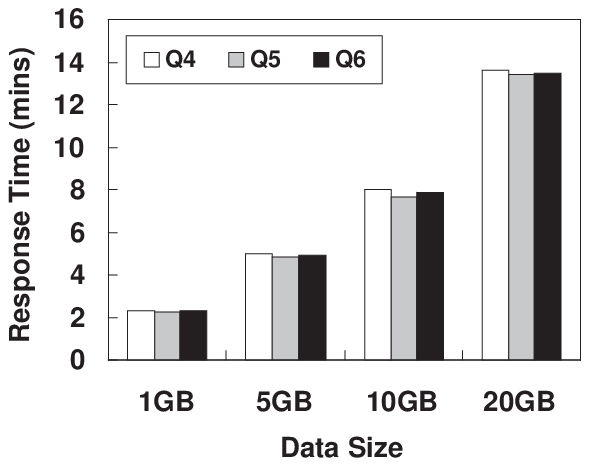, scale=0.85}
  \caption{Response Time of Chain-Type Queries}
  \label{fig:chaintime} 
    \end{center}  
\end{minipage}
\hfill
\begin{minipage}[t]{0.3\textwidth}
\begin{center}
 \epsfig{file=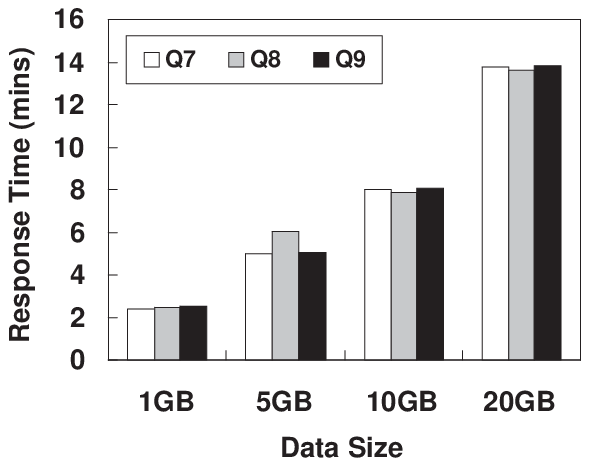, scale=0.85}
  \caption{Response Time of Mix-Type Queries}
  \label{fig:mixtime} 
      \end{center} 
\end{minipage}
\hfill
\\
\hfill
\begin{minipage}[t]{.3\textwidth}
\begin{center}
 \epsfig{file=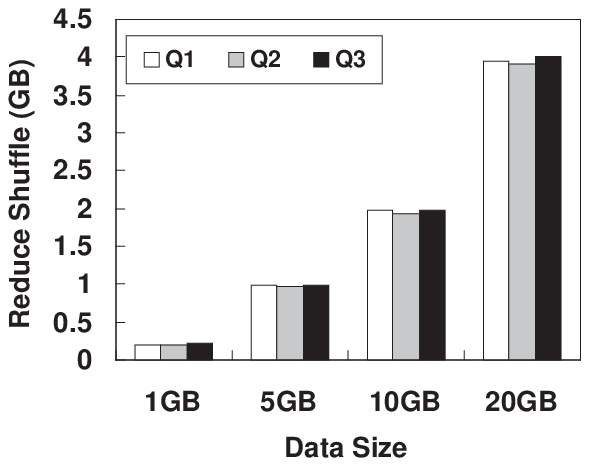, scale=0.85} 
  \caption{Reduce Shuffle of Star-Type Queries}\label{fig:starshuffle}
  \end{center}   
\end{minipage}
\hfill 
\begin{minipage}[t]{0.3\textwidth}
\begin{center}
 \epsfig{file=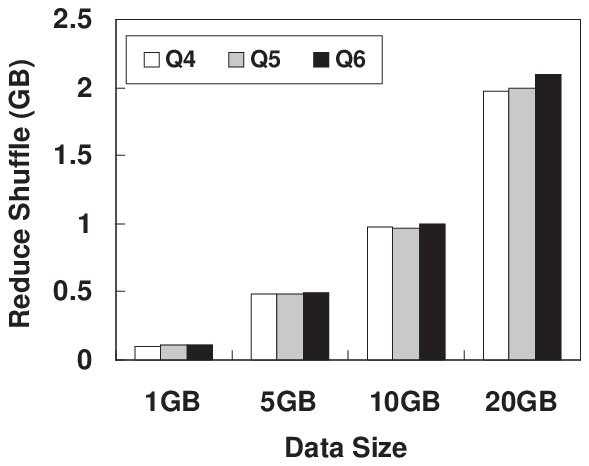, scale=0.85} 
  \caption{Reduce Shuffle of Chain-Type Queries}\label{fig:chainshuffle} 
    \end{center}  
\end{minipage}
\hfill
\begin{minipage}[t]{0.3\textwidth}
\begin{center}
 \epsfig{file=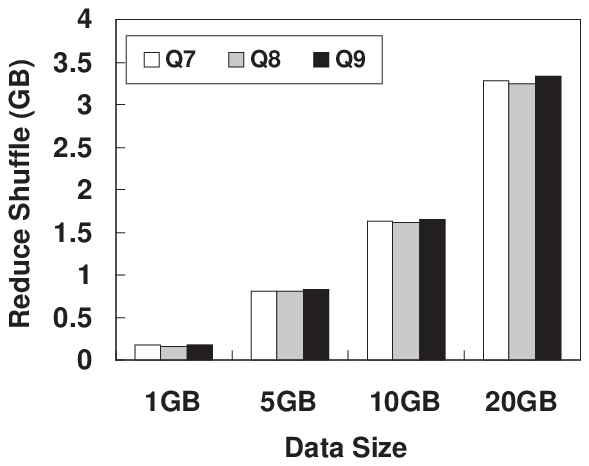, scale=0.85} 
  \caption{Reduce Shuffle of Mix-Type Queries}\label{fig:mixshuffle} 
     \end{center} 
\end{minipage}
\hfill
\end{figure*}


\subsection{Response Time}



Figure~\ref{fig:startime}-Figure~\ref{fig:mixtime} provide the response time of FCT search when we process the selected keyword queries over the TPC-H dataset with different sizes: 1GB, 5GB, 10GB and 20GB, respectively. 
For processing the query batches Q$_4$-Q$_9$, we first merge the two relations CUSTOMER and ORDERS, and then run the multiway-based MapReduce join by taking the LINEITEM as the fact relation.

From the experimental results, we find that most of time is spent on the MapReduce$^{1st}$ stage. For example, for the query batch Q$_1$, the MapReduce$^{1st}$ stage consumes 1.25 mins while the MapReduce$^{2nd}$ stage takes 0.6 mins for 1GB dataset; 
the first stage consumes 3.25 mins while the second takes 0.8 mins for 5GB dataset; 
the first stage consumes 6.1 mins while the second takes 0.81 mins for 10GB dataset; 
the first stage consumes 11.33 mins while the second takes 0.93 mins for 20GB dataset;
For other query batches, we can get the similar observations that the first stage takes the high percentage of the total response time. 
In addition, from the experimental results, we can find that map() in MapReduce$^{1st}$ stage takes about 0.33 mins to load a block with size of 64MB and the loading balance can be guaranteed by splitting the dataset into multiple blocks.

For instance, consider Q$_1$, 5GB dataset and 8 mappers, it is splitted into 68 map tasks and each mapper approximately load 8 number of blocks. As such, the map stage may take about 0.33*8 = 2.64 mins to finish all mappers' workloads. At the reduce stage, the shuffle() takes high time cost than sort() and reduce(), e.g., shuffling spends 1.91 mins while sorting takes 0.1 mins and reduce() takes 0.43 mins for one reducer in processing 5GB dataset using 8 reducers. Fortunately, the shuffling can be processed in parallel at the map stage. Based on this, we can find that the response time can be minimized to max\{2.64, 1.91\} + 0.1 + 0.43 = 3.17 mins at most with regards to the 5GB dataset and 8 worker nodes.   

From the result analysis, we can get that the total response time constrains to the maximal value of loading time and the shuffling time. To reduce the loading time, we can add more worker nodes into the cluster. 
To reduce the shuffling time, we send each fact tuple into one reduce task and copy the required dimension tuples into their corresponding reduce tasks based on our proposed scheduling strategy, which can reduce the shuffling operation times because generally fact relation is much larger than dimension relations.


\subsection{Reduce Shuffle Size}




Figure~\ref{fig:starshuffle}-Figure~\ref{fig:mixshuffle} show the reduce shuffle space usage when we process the selected keyword queries over the TPC-H dataset with different sizes: 1GB, 5GB, 10GB and 20GB, respectively. From Figure~\ref{fig:starshuffle}, we find that the shuffle space usage approximately takes 20\% of the dataset size. From Figure~\ref{fig:chainshuffle}, we find that the shuffle space usage approximately takes 10\% of the dataset size. From Figure~\ref{fig:mixshuffle}, we find that the shuffle space usage approximately takes 17\% of the dataset size. This is because for Chain-Type and Mix-Type queries, the ORDERS relation can be reduced by joining with the CUSTOMER relation, which can reduce the number of ORDERS tuples to involve in the multiway-join of MapReduce. Particularly, for Chain-Type queries, its multiway-join uses two attributes as the composite key, e.g., suppkey and orderkey in Figure~\ref{fig:querytype}. In this case, the number of copies for a dimension tuple is much smaller than that of Star-Type taking three attributes as the composite key.    

From experimental results of Q$_1$ over 10GB dataset, we find that the shuffle space usages of the 8 reducers are unbalanced. For half of the reducers, their individual shuffle space cost is approximately 344MB, in which the number of reduce input records is 10,843,452. While for the other four reducers, their individual shuffle space cost is approximately 144MB, in which the number of reduce input records is 4, 070, 586. 
From the result analysis, we can get that the unbalanced shuffling often happens when we deal with big data, which  may affect the total performance greatly. This is also the reason that the shuffling time cost takes the high percentage of the time cost of reduce stage at MapReduce.
 
\subsection{Verifying Uneven Distribution-based Shuffling Strategy}

\begin{figure}[htbp]
  \centering  
    \includegraphics[scale=0.85]{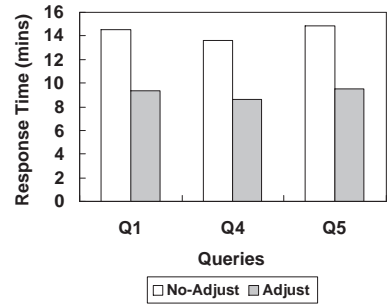}    
  \caption{Response Time of Processing Uneven Data Distribution}\label{fig:uneventime} 
\end{figure}

\begin{figure}[htbp]
  \centering
     \includegraphics[scale=0.85]{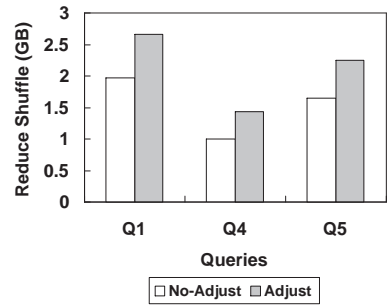} 
  \caption{Reduce Shuffle Cost of Processing Uneven Data Distribution}\label{fig:unevenshuffle} 
\end{figure}


In order to verify the uneven distribution-based shuffling strategy, we modify the 10GB data set into a dataset with much higher data skew by removing some tuples and repeatedly adding some other tuples. To make a tradeoff between the estimated precision of the sampling and its performance, we set a medium number of reduce tasks, e.g., 27 reduce tasks in this experiments. This is because generating a large number of reduce tasks may balance the workloads of the reducers as much as possible, while it may consume more time to do the samples for the big number of reduce tasks. Since we have 8 available reducers, each reducers may process about 3 reduce tasks if these reduce tasks are partitioned in balance. To illustrate the effect of balancing the workloads of reducers, we compare the performance of the query batches Q$_1$, Q$_4$ and Q$_5$ over the 10GB dataset when the number of reduce tasks are set as 8 and 27, respectively. And we still use 8 reducers (worker nodes).  

Figure~\ref{fig:uneventime} and Figure~\ref{fig:unevenshuffle} show the response time and the shuffle space usage, respectively. The label \textit{No-Adjust} corresponds to the case of  generating 8 reduce tasks while the label \textit{Adjust} specifies the case of having 27 reduce tasks. By adjusting the workloads of reducers, the response time can be reduced by about 36.03\% as shown in Figure~\ref{fig:uneventime} because the hot reducers can be alleviated. However, the shuffle space usage may be increased by about 38.16\% because a dimension tuple may be copied into more reduce tasks. Although the data size of shuffle is increased, many data can be compressed before the reducers pull because the reduce tasks to be sent to the same reducers may include more duplicates of dimension tuples. 









\section{Related Work}\label{sec:relatedwork}

 
\subsection{Keyword Query}
\cite{DBLP:conf/icde/YaoCHH12} addressed the problem of keyword query reformulation
in the structured data. These reformulated queries provided alternative descriptions
of an original keyword query. To do this, they first extracted the term relations in an
offline mode and then generated a set of new queries.
\cite{DBLP:journals/pvldb/SarkasBDK09} proposed a search model, similar in spirit to 
faceted search, that enables the progressive refinement of a keyword query result. The 
refinement process was driven by suggesting interesting expansions of the original query 
with additional search terms.
\cite{DBLP:conf/sigmod/ChuBCDN09} proposed to take as input a target database and then generated and indexed a set of query forms offline. At query time, a user with a question to be answered issued standard keyword search queries; but instead of returning tuples, the system returned forms relevant to the question. The user may then build a structured query with one of these forms and submit it back to the system for evaluation.
\cite{DBLP:journals/pvldb/PuY08,DBLP:conf/icde/LuWLL11} introduced the problem of query cleaning for keyword search queries in a database context and proposed a set
of effective and efficient solutions. Our FCT search approach can be taken as the supplementary tool to improve the effectiveness and efficiency of the above works when they process big data. 


\subsection{Join Operations in MapReduce}

Repartition Join \cite{DBLP:conf/sigmod/BlanasPERST10} is a two-way based join strategy, which is the most commonly used join strategy
in the MapReduce framework. In this join strategy, L
and R are dynamically partitioned on the join key and the
corresponding pairs of partitions are joined. The repartition join is used in our experiments when we join the relations CUSTOMER and ORDERS. Variants of the standard
repartition join are used in Pig \cite{DBLP:conf/sigmod/OlstonRSKT08}, Hive \cite{web:apachehive} and Jaql \cite{web:jaql} today. 
Another two-way based join in MapReduce is the Broadcast Join \cite{DBLP:conf/sigmod/BlanasPERST10} where the reference table R is much smaller
than the log table L, i.e. $|R|$ $<<$ $|L|$. Instead of moving
both R and L across the network as in the repartition-based
joins, it can simply broadcast the smaller table R, as it
avoids sorting on both tables and more importantly avoids
the network overhead for moving the larger table L.
However, two-way based join strategies are not suitable to process multiway join in MapReduce because it will run more MapReduce jobs.

Besides our adopted lagrangean multipliers based mutliway join strategy \cite{DBLP:conf/edbt/AfratiU10,DBLP:journals/tkde/AfratiU11}, there is another heauristic multiway join strategy in \cite{DBLP:journals/tkde/JiangTC11} that also focused on the join strategy that a dataset (we denoted it as fact dataset)
has more than one join columns with the other datasets (we denote as dimension datasets). 
Similar to \cite{DBLP:conf/edbt/AfratiU10}, 
\cite{DBLP:journals/tkde/JiangTC11} designed the partition function over dimension datasets and then the fact dataset can be splitted based on the partition functions of dimension datasets.
Differently, \cite{DBLP:journals/tkde/JiangTC11} couldn't determine the optimal number of partitions for each dimension datasets, and just proposed a heuristic approach to solve
the optimization problem. However, from  \cite{DBLP:conf/edbt/AfratiU10,DBLP:journals/tkde/AfratiU11}, we can derive the optimal number of each dimension datasets to be splitted.

\section{Conclusions}\label{sec:conclusions}
In this paper, we studied the problem of query-driven FCT search over big data in parallel using the MapReduce framework. We proposed a MapReduce-based FCT search approach that  consists of two MapReduce jobs: the first is to calculate the statistical information of the query over the big data while the second is to compute the term frequencies.  
In addition, we showed how to partition the data across worker nodes in order to balance their workloads when data distribution is uniformed or skewed. 
We also described the detailed procedures of the two MapReduce jobs and their implementation algorithms. 
At last, the MapReduce-based FCT search approach is implemented over our university cloud platform \textit{SwinCloud}. We conducted the experiments over the TPC-H benchmark datasets with different sizes and data distributions. From the experimental analysis, we concluded that the main time cost of our approach were spent on the data loading and intermediate data shuffling, which can be improved by adding more worker nodes and copying the dimension tuples in an optimal way.


\bibliographystyle{IEEEtran}

\bibliography{KS,MR}

\end{document}